\documentclass[12pt,oneside,openany]{article}
\usepackage{amsfonts,amssymb,graphicx}
\usepackage{cite}
\usepackage[affil-it]{authblk}
\usepackage{indentfirst}
\usepackage{newlfont}
\usepackage{amsthm}
\usepackage{amsmath}
\usepackage{latexsym}
\usepackage{eucal}
\usepackage{amssymb}
\usepackage{mathrsfs}
\usepackage{graphicx}

\newtheorem{teorema}{Theorem}

\newtheorem{lemma}{Lemma}

\DeclareMathOperator{\sgn}{sgn}
\DeclareMathOperator{\diag}{diag}
\begin{document}
\title{Rescaled Magnetization\\ for Critical Bipartite Mean-Fields Models}

\author{Micaela Fedele}
\affil{Courant Institute of Mathematical Sciences, New York University}
\date{}

\maketitle

\begin{abstract}
We consider a bipartite generalization of the Curie-Weiss model in a critical regime. In order to study the asymptotic behavior of the 
random vector of the 
total magnetization we apply the change of variables that diagonalizes the Hessian matrix of the pressure functional 
associated to the model. We obtain a new vector that, suitably 
rescaled, weakly converges to the product of a Gaussian distribution and a distribution proportional to $\exp(-\xi x^{4})$, where the 
positive constant $\xi$ can be computed from the pressure functional. 
\end{abstract}
\vspace{1cm}
\section*{Introduction}
The standard statistical mechanics approach to study the phase transitions of a
model amounts to analyze its pressure functional looking for points of non-analiticity. A different and interesting way to achieve this result is provided by the description of the asymptotic behavior of the sum of the
random variables occurring in the model \cite{kadanoff1976scaling}. Due to the interaction, these variables do not satisfy the hypothesis of 
independence required in order to apply the central limit theorem. Nevertheless, if the phase of the model is not critical it is expected that their sum, with
square root normalization, shows a central limit type behaviour and thus converges toward a Gaussian distribution. 

For a wide class of models known in literature as Curie-Weiss models \cite{kac1969mathematical, thompson1988classical,ellis2005entropy} such a prediction was 
confirmed in \cite{ellis1978limit,ellis1978statistics,ellis1980limit} where it was also described the asymptotic non-Gaussian behavior corresponding to the 
critical point of the model (this result was first shown in \cite{simon1973varphi4}). In \cite{fedele2011scaling}, the same analysis was conducted for the multi-species mean-field model: a generalization of the Curie-Weiss model in which 
spin random variables are partitioned into an arbitrary number of groups and both the interaction and the external field parameters take 
different values only depending on the groups variables belong to. 

This model, whose bipartite symmetric version was introduced in the 50s to reproduce the phase transition of the so-called metamagnets 
\cite{gorter1956transitions, motizuki1959metamagnetism, bidaux1967antiferromagnetisme, kincaid1975phase, galam1980new}, recently has been receiving a 
renewed attention \cite{gallo2008bipartite, barra2011equilibrium, fedele2012rigorous, fedele2013inverse} mainly thanks to its potential ability to account for 
the collective behavior of socio-economic agents \cite{contucci2007modeling, contucci2008phase, gallo2008parameter, agliari2010new, barra2012statistical}. 
The idea of using Statistical Mechanics to describe the outcomes of individual decisions at population level appeared in literature in the early 80s \cite{galam1982sociophysics} as a consequence of the increased emphasis on the role played by social interaction in shaping personal preferences. In fact, the occurence of sudden behavioral shifts such as trends, fads and crashes might be hardly understood if the agents would take decisions without influencing one another. Interestingly, the need to incorporate peer-to-peer effects into the framework of the Discrete Choice theory \cite{mcfadden2001economic}, a model able to forecast with a remarkable agreement collective phenomena in which social interactions do not play a substantial role, led to the formulation of a model \cite{durlauf1999can, brock2001discrete} equivalent from the mathematical point of view to the Curie-Weiss model. Since the Discrete Choice theory rephrased as a statistical mechanics model corresponds to a mixture of a finite number of discrete perfect gases, its natural extension to the interacting case is represented by the multi-species mean-field model. 

Despite the great importance that this model may play toward the understanding of socio-economical phenomena, a complete description of 
its phase space is 
lacking to this day. On one hand standard investigations of the critical points of the pressure functional were performed only in 
specific cases \cite{fedele2012rigorous, barra2011equilibrium}, on the other the analysis of the asymptoic behavior of the sums of the 
spins \cite{fedele2011scaling} was done under the assumptions that the Hamiltonian is a convex function of the 
the sums of the spins of each group (convexity hypothesis), and the pressure functional can be written as an homogeneous and 
strictly positive polynomial around its minimum points (homogeneity hypothesis).

In this paper we made a step forward in filling the gap in literature dealing with a specific situation beyond the homogeneity hypotesis for the 
bipartite mean-field model in absence of the external field. In particular we consider the model as the unique minimum point of the 
pressure functional is the origin and its Hessian matrix computed in that point has determinant equal to zero without being equal to the 
null matrix. By applying to the vector of the total magnetizations the orthogonal matrix that 
diagonalizes the Hessian matrix of the pressure functional, we obtain a new 
random vector that, properly rescaled, weakly converges to the product of a Gaussian distribution and a distribution proportional to 
$\exp(-\xi x^{4})$, where the positive constant $\xi$ can be computed from the pressure functional.
The discovery of a non-central limit type behavior allows us to assert that the bipartite mean-field model 
undergoes a phase transition in the considered scenario, as previously proved only when the two groups of particles had the same 
size and the same strenght of internal interaction \cite{fedele2012rigorous}. 

This paper is organized as follows. Section one describes the model and states the main result. Section two contains the proof of the main result. The appendix presents the proof of the lemmas used to prove the main result.

\section{Definitions and Statement}
\noindent We consider a system of $N$ spin particles divided in two subsets $P_{1}$ and $P_{2}$, respectively of size $N_{1}$ and $N_{2}$, such that 
$P_{1}\cap P_{2}=\emptyset$ and $N_{1}+N_{2}=N$.
Particles interact with each other according to the following Hamiltonian: 
\begin{equation}\label{Hamiltonian.1}
H_{N}(\boldsymbol{\sigma})=-\frac{1}{2N}\sum_{i,j=1}^{N}J_{ij}\sigma_{i}\sigma_{j}
\end{equation}
where $\sigma_{i}$ represents the spin of the particle $i$ and $J_{ij}$ is the parameter that tunes the mutual 
interaction between the particles $i$ and $j$. Such a parameter takes values according to the following symmetric matrix:
\vspace{-0.3cm}
\begin{displaymath}
         \begin{array}{ll}
                \\
                P_1 \left\{ \begin{array}{ll||}
                                      \\
                                   \end{array}  \right.
                                        \\
                P_2 \left\{ \begin{array}{ll||}
                                        \\
                                   \end{array}  \right.

         \end{array}
          \!\!\!\!\!\!\!\!
         \begin{array}{ll||}
                \quad
                 \overbrace{\qquad }^{\textrm{$P_1$}}
                 \overbrace{\qquad }^{\textrm{$P_2$}}
                
                  \\
                 \left(\begin{array}{c|c}
                               \mathbf{ J}_{11}  &  \mathbf{ J}_{12} \\
                                 \hline
                              \mathbf{ J}_{12}^{T} & \mathbf{ J}_{22} \\
                             
                      \end{array}\right)
               \end{array}
\end{displaymath}\\
\noindent where each block $\mathbf{ J}_{ls}$ has constant elements $ J_{ls}$. We assume $J_{11}$ and $J_{22}$ be strictly positive, 
while $J_{12}$ can be positive or negative allowing both ferromagnetic and antiferromagnetic interactions. We observe that for $J_{12}=0$
the bipartite mean-field model degenerates toward two distinct Curie-Weiss models.

By introducing the total magnetization of each group:
\begin{equation*}
S_{1}(\boldsymbol{\sigma})=\sum_{i \in P_{1}}\sigma_{i}\quad\qquad S_{2}(\boldsymbol{\sigma})=\sum_{i \in P_{2}}\sigma_{i}
\end{equation*}
we may easily express the Hamiltonian (\ref{Hamiltonian.1}) as a binary quadratic form:
\begin{equation}\label{Hamiltonian.2}
 H_{N}(\boldsymbol{\sigma})=-\frac{1}{2N}\langle \mathbf{J}\mathbf{S},\mathbf{S}\rangle
\end{equation}
where $\mathbf{S}=(S_{1}(\boldsymbol{\sigma}), S_{2}(\boldsymbol{\sigma}))$ is the vector of the total magnetizations and
\begin{equation*} 
\mathbf{J}=\begin{pmatrix}
   J_{11} & J_{12}\\
J_{12} & J_{22}
  \end{pmatrix}
\end{equation*}
is the so-called reduced interaction matrix.
The joint distribution of a spin configuration $\boldsymbol{\sigma}=(\sigma_{1},\dots ,\sigma_{N})$ is given by the Boltzmann-Gibbs 
measure:
\begin{equation}\label{measure.BG}
P_{N,\mathbf{J}}\{\boldsymbol{\sigma}\}=Z_{N}^{-1}\exp(-H_{N}(\boldsymbol{\sigma}))\prod\limits_{i=1}^{N}
d\rho(\sigma_{i})
\end{equation}
where $Z_{N}$ is the partition function
\begin{equation*}\label{partition}
Z_{N}=\int_{\mathbb{R}^{N}}\exp(-H_{N}(\boldsymbol{\sigma}))\prod\limits_{i=1}^{N}d\rho(\sigma_{i})
\end{equation*}
and $\rho$ is the measure:
\begin{equation*}\label{measure.ro}
\rho(x)=\frac{1}{2}\Big(\delta(x-1)+\delta(x+1)\Big)
\end{equation*}
where $\delta(x-x_{0})$ with $x_{0}\in\mathbb{R}$ denotes the unit point mass with support at $x_{0}$. The definition of $\rho$ implies 
that each spin variable can take only the values $\pm 1$. The inverse temperature parameter $\beta$ is not explicitly written because we 
consider it absorbed within the model parameters.

The existence of the thermodynamic limit of the pressure $p_{N}=N^{-1}\ln Z_{N}$ associated to the model is proved in 
\cite{gallo2008bipartite} where it is also computed the exact value of such a limit for models whose Hamiltonian is a convex function of 
the total magnetizations (for the computation of the limit in the general case see \cite{fedele2012rigorous}). It holds:
\begin{equation*}
 \lim_{N\rightarrow\infty}p_{N}=\ln 2-\inf\{G(x_{1},x_{2}):(x_{1},x_{2})\in[-1,1]^{2}\}
\end{equation*}
where the pressure functional $G$ is:
\begin{align}\label{pressure.functional}
G(x_{1},x_{2}) & = \frac{1}{2}\left(\alpha^{2}J_{11}x_{1}^{2}+2\alpha(1-\alpha)J_{12}x_{1}x_{2}+(1-\alpha)^{2}J_{22}x_{2}^{2}\right)\nonumber\\
&\quad-\alpha\ln\cosh(\alpha J_{11}x_{1}+(1-\alpha)J_{12}x_{2})\nonumber\\
&\quad-(1-\alpha)\ln\cosh(\alpha J_{12}x_{1} +(1-\alpha)J_{22}x_{2})
\end{align}
with $\alpha=N_{1}/N$ the relative size of the subset $P_{1}$.

In \cite{fedele2011scaling}, it is shown the basic role played by the functional $G$ in determining the limiting behavior of the random 
vector $\mathbf{S}$. We recall briefly the results, obtained under the convexity and the homogenity hypothesis described in the 
introduction.
When $G$ has a unique minimum point, the random vector $\mathbf{S}$, suitably rescaled, weakly converges to a bivariate Gaussian only if the order of the approximation of $G$ around that point is the second. Otherwise, $\mathbf{S}$ converges to a 
distribution proportional to $\exp(-\bar{P_{k}}(x_{1},x_{2}))$ where $\bar{P_{k}}(x_{1},x_{2})$ is the homogeneous and strictly positive 
polynomium of order $k>2$ that approximates $G$ around the minimum point. When there are more minimum points, analogous results are valid locally around each of them. 

In this paper, we consider the bipartite mean field model defined by the Hamiltonian (\ref{Hamiltonian.2}) as
the pressure functional $G$ has a unique minimum point, the origin, in which the determinant of its Hessian matrix is equal to zero and 
the convexity hypothesis is still verified. In the considered case, the homogeneity hypothesis is not true unless all the elements 
of the Hessian matrix of $G$ at the origin are equal to zero. Due to the convexity hypothesis that happens if and only if $J_{12}=0$. 
The asymptotic behavior of the vector of the total magnetizations $\mathbf{S}$ as the parameter $J_{12}$ assumes values different from 
zero is investigated in the following:

\begin{teorema}\label{teo}
	Consider the bipartite mean-field model described by the Hamiltonian (\ref{Hamiltonian.2}) where the matrix $\mathbf{J}$ is positive definite with 
	$J_{12}\neq 0$ and let the origin be the unique minimum point of the pressure functional $G$ given by (\ref{pressure.functional}). 
	Denoted by $\lambda_{M}$ and $\lambda_{m}$ respectively the largest and the smallest eigenvalue of the Hessian matrix of $G$ computed in the origin and by
	$\mathbf{v}_{M}$ and $\mathbf{v}_{m}$ the corresponding eingenvectors, define the random vector $\mathbf{\widetilde{S}}=(\widetilde{S}_{1}(\boldsymbol{\sigma}),\widetilde{S}_{2}(\boldsymbol{\sigma}))$ as:
	\begin{equation}\label{definizione.S.tilde}
	\mathbf{\widetilde{S}}=(\mathbf{A}^{2})^{-1}\mathbf{P}\mathbf{A}^{2}\mathbf{S}
	\end{equation} 
	 where $\mathbf{A}=\diag\{\sqrt{\alpha},\sqrt{1-\alpha}\}$ and 
	$\mathbf{P}=(\mathbf{v}_{M}||\mathbf{v}_{M}||^{-1};\mathbf{v}_{m}||\mathbf{v}_{m}||^{-1})$. If $\lambda_{m}=0$ then there exist $\xi_{1},\xi_{2}\in\mathbb{R}$
	strictly positive such that, as $N\rightarrow\infty$, the random vector 
	\begin{equation}\label{vettore.S.1}
	\left(\dfrac{\widetilde{S}_{1}(\boldsymbol{\sigma})}{(N_{1})^{1/2}},\dfrac{\widetilde{S}_{2}(\boldsymbol{\sigma})}{(N_{2})^{3/4}}\right)
	\end{equation}
	weakly converges to 
	\begin{equation}\label{result.theorem}
	\dfrac{\exp\left(-\xi_{1}x_{1}^{2}-\xi_{2}x_{2}^{4}\right)dx_{1}dx_{2}}{\displaystyle{\int_{\mathbb{R}^{2}}}\exp\left(-\xi_{1}x_{1}^{2}-\xi_{2}x_{2}^{4}\right) 
	dx_{1}dx_{2}}.
	\end{equation}
\end{teorema}
\vspace{0.7cm}
\noindent We claim that the coefficients $\xi_{1}$ and $\xi_{2}$ are related to the pressure functional $G$ and will be computed explicitly in the proof.
\section{Proof of the Statement}
\noindent Let us start by determining for which values of the model parameters the origin is the unique minimum point of the function 
$G$, that is the unique solution of the system: 
\begin{equation}\label{mean.field.equations}
\begin{cases}
x_{1} &\!\!\!\!= \tanh\left(\alpha J_{11}x_{1}+(1-\alpha)J_{12}x_{2}\right) \\
x_{2} &\!\!\!\!=\tanh\left(\alpha J_{12}x_{1}+(1-\alpha)J_{22}x_{2}\right)
\end{cases} 
\end{equation}
that represents the estremality conditions of $G$. Since $J_{12}\neq 0$, after inverting the hyperbolic tangent in the two equations, we 
can rewrite the system (\ref{mean.field.equations}) in the following fashion: 
\begin{equation}\label{sistema.invertito}
\begin{cases}
x_{2} &\!\!\!\!=\dfrac{1}{(1-\alpha)J_{12}}\left(\tanh^{-1}(x_{1})-\alpha J_{11}x_{1}\right)\\
x_{1} &\!\!\!\!= \dfrac{1}{\alpha J_{12}}\left(\tanh^{-1}(x_{2})-(1-\alpha)J_{22}x_{2}\right)
\end{cases} 
\end{equation}
that lends itself to a graphic resolution. Considered the Cartesian coordinate system $x_1x_2$, defined the functions
\begin{align*}
f_{1}(x_1)&=\dfrac{1}{(1-\alpha)J_{12}}\left(\tanh^{-1}(x_{1})-\alpha J_{11}x_{1}\right),\\
f_{2}(x_1)&=\dfrac{1}{\alpha J_{12}}\left(\tanh^{-1}(x_{1})-(1-\alpha)J_{22}x_{1}\right),
\end{align*}
and denoted by $\gamma_{1}$ and $\gamma_{2}$ respectively the graph of $f_{1}$ and $f_{2}$, the solutions of (\ref{sistema.invertito}) 
are the intersections between $\gamma_{1}$ and the symmetrical curve of $\gamma_{2}$ with respect to the line $x_2=x_1$ (in the following
we denote the latter curve by $\hat{\gamma}_{2}$). 
\begin{figure}[!h]
\centering
\includegraphics[scale=0.35]{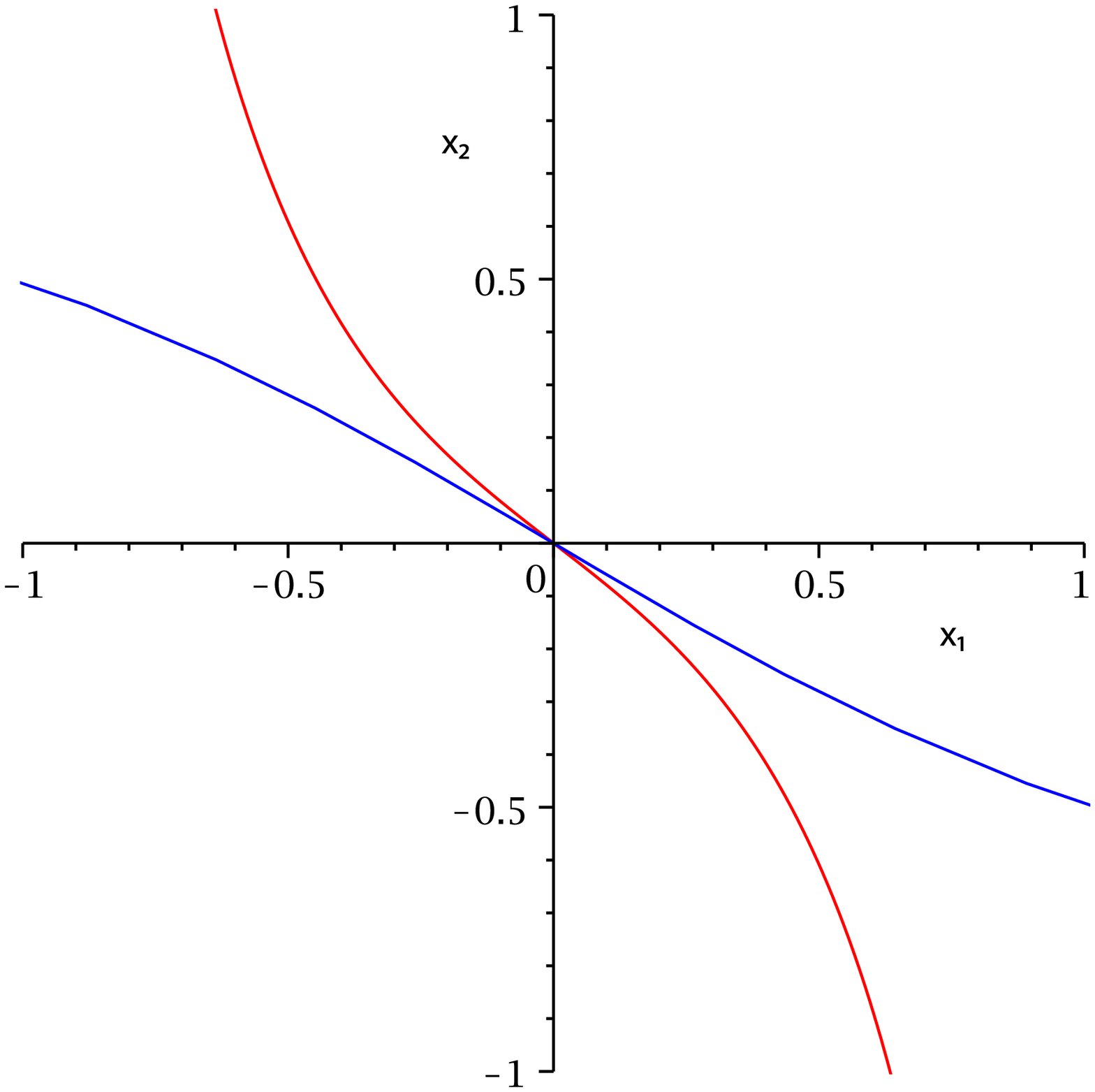} \qquad
\includegraphics[scale=0.35]{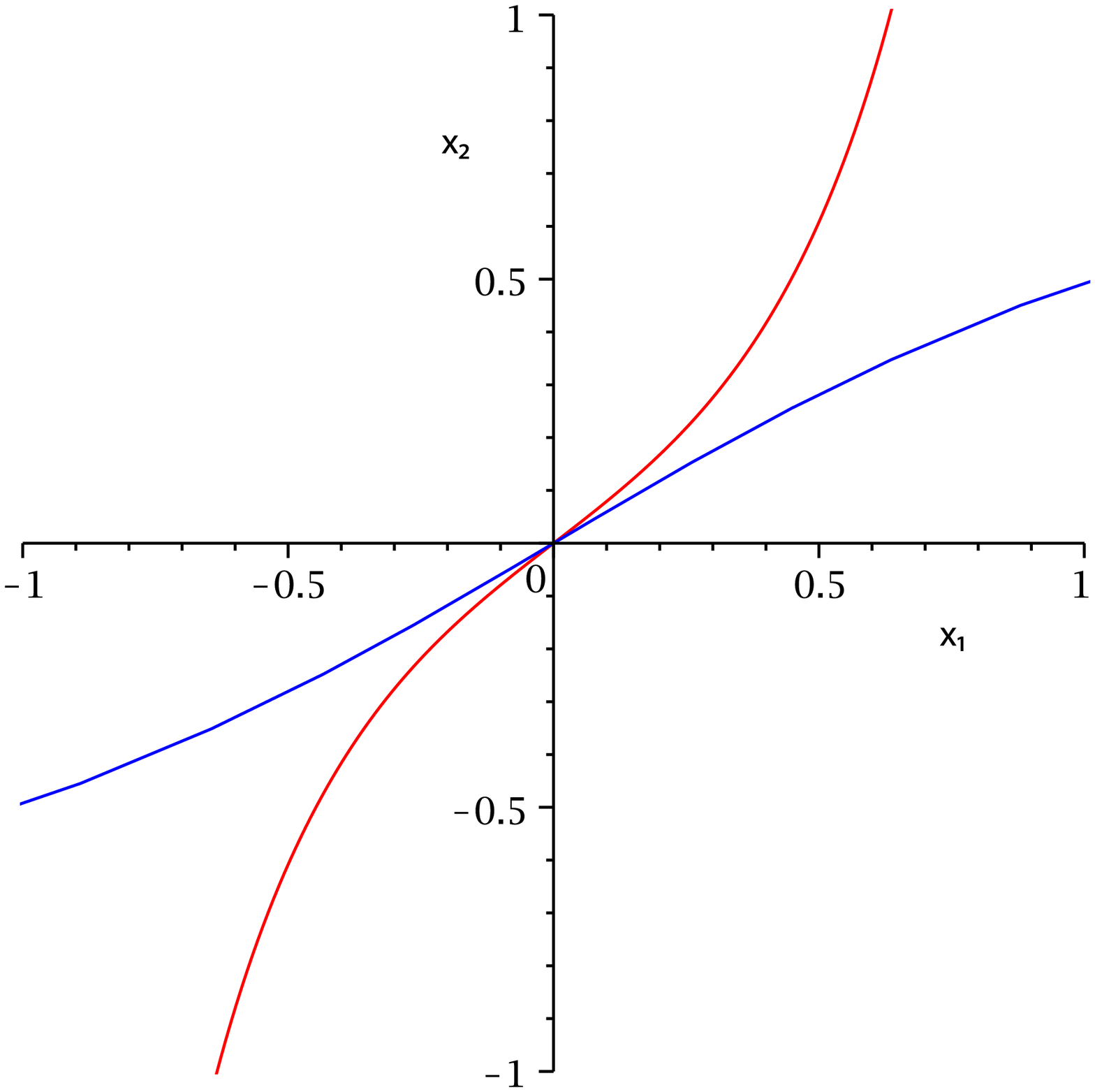}
\caption{Graphic rapresentation of the system (\ref{sistema.invertito}) in the case of a unique solution. The red unlimited curve is 
$\gamma_{1}$ while the blue limited curve is $\hat{\gamma}_{2}$. In the left panel $J_{12}<0$, in the right panel $J_{12}>0$.}
\label{fig}
\end{figure}

The functions $f_1$ and $f_2$ have no inflection points over the origin. Therefore, for a unique intesection of $\gamma_{1}$ and 
$\hat{\gamma}_{2}$ in the origin, $f_1$ and $f_2$ both
must be strictly increasing or decreasing and the slope in the origin to $\gamma_1$ must be bigger in absolute value
than the slope in the origin to $\hat{\gamma}_{2}$. By computing the first derivatives of $f_{1}$ and $f_{2}$ it is easy to show that the condition on the 
monotonicity is fulfilled as 
$J_{11}<\alpha^{-1}$ and $J_{22}<(1-\alpha)^{-1}$, while the other condition is true as 
\begin{equation}\label{derivatecurve}
 (1-\alpha J_{11})(1-(1-\alpha)J_{22})\geq \alpha(1-\alpha)J_{12}^{2}\;.
\end{equation}
In particular, considered the elements of the Hessian matrix of $G$ in the origin
\begin{align}\label{elementi.Hessiana}
 \mathcal{H}_{11} &= \alpha^{2}J_{11}(1-\alpha J_{11})-\alpha^{2}(1-\alpha)J_{12}^{2},\nonumber\\
 \mathcal{H}_{12} &= \alpha(1-\alpha)J_{12}(1-\alpha J_{11}-(1-\alpha)J_{22}),\\
\mathcal{H}_{22} &= (1-\alpha)^{2}J_{22}(1-(1-\alpha)J_{22})-\alpha(1-\alpha)^{2}J_{12}^{2},\nonumber
\end{align}
when the inequality (\ref{derivatecurve}) is verified as an identity we have:
\begin{equation*}
 \lambda_{m}= \frac{1}{2}\left(\mathcal{H}_{11}+\mathcal{H}_{22}-\sqrt{(\mathcal{H}_{11}-\mathcal{H}_{22})^{2}+4\mathcal{H}_{12}^{2}}
\right)=0.
\end{equation*}
Therefore the hypotheses of the theorem are true when the model parameters verify the following system of conditions:
\begin{equation}\label{00.minimum.point}
\begin{cases}
J_{12}\neq 0,\\
J_{11}< \alpha^{-1},\\
 J_{22}< (1-\alpha)^{-1},\\
(1-\alpha J_{11})(1-(1-\alpha)J_{22})= \alpha(1-\alpha)J_{12}^{2},\\
\alpha J_{11}+(1-\alpha)J_{22}-1>0; 
\end{cases}
\end{equation}
whose latter inequality is a direct consequence of the positive definiteness of the matrix $\mathbf{J}$ when the
(\ref{derivatecurve}) is verified as an identity.

To prove the theorem, considered the orthogonal matrix 
$\mathbf{P}=(\mathbf{v}_{M}||\mathbf{v}_{M}||^{-1};\mathbf{v}_{m}||\mathbf{v}_{m}||^{-1})$ of the normalized 
eingenvectors of the Hessian matrix of $G$ in the origin, we introduce the function 
$\widetilde{G}(\mathbf{x})=G(\mathbf{P}^{-1}\mathbf{x})$, that will play the same role of the function $G$ in \cite{fedele2011scaling}.
 Although the result (\ref{result.theorem}) holds for any possible choice of the eigenvectors, to make 
the proof clear it is worth to work with the explicit expressions of two of them. By choosing 
\begin{equation}\label{eigenvectors}
\mathbf{v}_{M}=\left(1;\frac{\lambda_{M}-\mathcal{H}_{11}}{\mathcal{H}_{12}}\right)\qquad 
\mathbf{v}_{m}=\left(1;\frac{\lambda_{m}-\mathcal{H}_{11}}{\mathcal{H}_{12}}\right)
\end{equation}
where $\mathcal{H}_{11}$, $\mathcal{H}_{12}$ and $\mathcal{H}_{22}$ are given by (\ref{elementi.Hessiana}) and considering the Euclidean norm, we have: 
\begin{align}\label{function.G.tilde}
\widetilde{G}(x_{1},x_{2})&=\frac{1}{2}\left(\widetilde{J}_{11} x_{1}^{2}+2\widetilde{J}_{12}x_{1}x_{2}+\widetilde{J}_{22}x_{2}^{2}
\right)-\alpha\ln\cosh(a_{1}x_{1}+a_{2}x_{2})\nonumber\\
&\quad-(1-\alpha)\ln\cosh(b_{1}x_{1}+b_{2}x_{2})
\end{align}
with:
\begin{align}\label{coefficenti.Jtilde}
 \widetilde{J}_{11}&=\dfrac{J_{11}\alpha^{2}\mathcal{H}_{12}^{2}\!+\! J_{22}(1\!-\!\alpha)^{2}(\lambda_{M}\!-\!\mathcal{H}_{11})^{2}\!+\!
2J_{12}\alpha(1\!-\!\alpha)\mathcal{H}_{12}(\lambda_{M}\!-\!\mathcal{H}_{11})}{\mathcal{H}_{12}^{2}+(\lambda_{M}\!-\!\mathcal{H}_{11})^{2}},\nonumber\\
\widetilde{J}_{22}&=\dfrac{J_{11}\alpha^{2}\mathcal{H}_{12}^{2}\!+\! J_{22}(1\!-\!\alpha)^{2}(\lambda_{m}\!-\!\mathcal{H}_{11})^{2}\!+\!
2J_{12}\alpha(1\!-\!\alpha)\mathcal{H}_{12}(\lambda_{m}\!-\!\mathcal{H}_{11})}{\mathcal{H}_{12}^{2}+(\lambda_{m}\!-\!\mathcal{H}_{11})^{2}},\\
\widetilde{J}_{12}&=\dfrac{J_{11}\alpha^{2}\mathcal{H}_{12}^{2}\!+\! J_{22}(1\!-\!\alpha)^{2}(\lambda_{M}\!-\!\mathcal{H}_{11})
(\lambda_{m}\!-\!\mathcal{H}_{11})\!+\!
J_{12}\alpha(1\!-\!\alpha)\mathcal{H}_{12}(\mathcal{H}_{22}\!-\!\mathcal{H}_{11})}{\sqrt{(\mathcal{H}_{12}^{2}+(\lambda_{M}\!-\!\mathcal{H}_{11})^{2})
(\mathcal{H}_{12}^{2}+(\lambda_{m}\!-\!\mathcal{H}_{11})^{2})}},\nonumber
\end{align}
and:
\begin{align}\label{coefficenti.abtilde}
a_{1}&=\dfrac{J_{11}\alpha|\mathcal{H}_{12}|+J_{12}(1-\alpha)\sgn(\mathcal{H}_{12})(\lambda_{M}-\mathcal{H}_{11})}
{\sqrt{\mathcal{H}_{12}^{2}+(\lambda_{M}-\mathcal{H}_{11})^{2}}},\nonumber\\
a_{2}&=\dfrac{J_{11}\alpha|\mathcal{H}_{12}|+J_{12}(1-\alpha)\sgn(\mathcal{H}_{12})(\lambda_{m}-\mathcal{H}_{11})}
{\sqrt{\mathcal{H}_{12}^{2}+(\lambda_{m}-\mathcal{H}_{11})^{2}}},\\
b_{1}&=\dfrac{J_{12}\alpha|\mathcal{H}_{12}|+J_{22}(1-\alpha)\sgn(\mathcal{H}_{12})(\lambda_{M}-\mathcal{H}_{11})}
{\sqrt{\mathcal{H}_{12}^{2}+(\lambda_{M}-\mathcal{H}_{11})^{2}}},\nonumber\\
b_{2}&=\dfrac{J_{12}\alpha|\mathcal{H}_{12}|+J_{22}(1-\alpha)\sgn(\mathcal{H}_{12})(\lambda_{m}-\mathcal{H}_{11})}
{\sqrt{\mathcal{H}_{12}^{2}+(\lambda_{m}-\mathcal{H}_{11})^{2}}}.\nonumber
\end{align}\\
\noindent The proof needs the following three lemmas.
\begin{lemma}\label{lemma.1}
	Suppose that for each $N$, $\mathbf{X}^{N}=(X^{N}_{1},X^{N}_{2})$ and $\mathbf{Y}^{N}=(Y^{N}_{1},Y^{N}_{2})$ are independent random vectors and 
	that $\mathbf{X}^{N}$ weakly converges to a distribution $\nu$ such that
	\begin{equation*}
	\int_{\mathbb{R}^{2}} e^{i\langle\mathbf{r},\mathbf{x}\rangle}d\nu(\mathbf{x})\neq 0 \quad\quad\text{for all}\;\;\mathbf{r}\in\mathbb{R}^{2}.
	\end{equation*}
	Then $\mathbf{Y}^{N}$ weakly converges to $\mu$ if and only if $\mathbf{X}^{N}+\mathbf{Y}^{N}$ weakly converges to the convolution $\nu *\mu$ 
	of the distributions $\nu$ and $\mu$.
\end{lemma}
\begin{proof} 
The result is a direct consequence of the equivalence between the weak convergence of measures and the pointwise convergence of 
characteristic functions (see \cite{durrett2004probability}).
\end{proof}

\begin{lemma}\label{lemma.2}
	If the matrix $\mathbf{J}$ of a model defined by the Hamiltonian (\ref{Hamiltonian.2}) is positive definite, then for any 
$N\in\mathbb{N}\setminus\{0\}$
	\begin{equation}\label{result.lemma.2}
	\int_{\mathbb{R}^{2}}\exp\left(\!-N\widetilde{G}(\mathbf{x})\right)d\mathbf{x}<\infty
	\end{equation}
	where the function $\widetilde{G}$ is given by (\ref{function.G.tilde}).
\end{lemma}
\noindent See appendix A for the proof.
\begin{lemma}\label{lemma.3}
	Denoted by $m=\min\{\widetilde{G}(\mathbf{x}):\mathbf{x}\in\mathbb{R}^{2}\}$, where $\widetilde{G}$ is defined in (\ref{function.G.tilde}), let $V$ be 
	any closed (possibly unbounded) subset of $\mathbb{R}^{2}$ which contains no global minima of the function $\widetilde{G}$. Then there exists 
	$\varepsilon>0$ such that
	\begin{equation}\label{result.lemma.3}
	e^{Nm}\int_{V} \exp(-N\widetilde{G}(\mathbf{x}))d\mathbf{x}=O(e^{-N\varepsilon}) \quad\quad\quad N\rightarrow\infty.
	\end{equation}
\end{lemma}
\noindent See appendix B for the proof.\vspace{0.3cm}\\
\indent Now we are ready to prove the theorem \ref{teo}. We will proceed in two steps. First, considered the random vector 
$(W_{1},W_{2})$ with joint distribution 
\begin{equation}\label{bivariate.Gaussian}
\dfrac{\sqrt{\det(\mathbf{A}\mathbf{\widetilde{J}}\mathbf{A}})}{2\pi}
\exp\left(-\frac{1}{2}\langle\mathbf{A}\mathbf{\widetilde{J}}\mathbf{A}\mathbf{w},\mathbf{w}\rangle\right)
\end{equation}
where $\mathbf{A}=\diag\{\sqrt{\alpha},\sqrt{1-\alpha}\}$ and $\mathbf{\widetilde{J}}$ is the 
symmetric $2\times 2$ matrix of 
elements $\widetilde{J}_{ij}$ defined in (\ref{coefficenti.Jtilde}), we show that, when $(W_{1},W_{2})$ is independent of $\mathbf{\widetilde{S}}$ for each $N$,
the distribution of the random vector
\begin{equation}\label{sum.vectors}
\left(\dfrac{\widetilde{S}_{1}(\boldsymbol{\sigma})}{(N_{1})^{1/2}},\dfrac{\widetilde{S}_{2}(\boldsymbol{\sigma})}{(N_{2})^{3/4}}\right)+
\left(W_{1},\dfrac{W_{2}}{(N_{2})^{1/4}}\right)
\end{equation}
is given by
\begin{equation}\label{step.1}
 \dfrac{\exp\left(-N\widetilde{G}\left(\dfrac{x_{1}}{(N_{1})^{1/2}},\dfrac{x_{2}}{(N_{2})^{1/4}}\right)\right)dx_{1}dx_{2}}
{\displaystyle{\int_{\mathbb{R}^{2}}}\exp\left(-N\widetilde{G}\left(\dfrac{x_{1}}{(N_{1})^{1/2}},\dfrac{x_{2}}{(N_{2})^{1/4}}\right)
\right)dx_{1}dx_{2}},
\end{equation}
that is a well defined distribution because the involved integral is finite by lemma (\ref{lemma.2}). Then we will analize the 
distribution (\ref{step.1}) as $N\rightarrow\infty$. 

\noindent Given $\theta_{1},\theta_{2}$ real
\begin{multline}\label{prob.1}
P\left\{\!W_{1}+\dfrac{\widetilde{S}_{1}(\boldsymbol{\sigma})}{(N_{1})^{1/2}}\!\leq\theta_{1},\!\dfrac{W_{2}}{(N_{2})^{1/4}}
+\dfrac{\widetilde{S}_{2}(\boldsymbol{\sigma})}{(N_{2})^{3/4}}\leq\theta_{2}\!\right\}\\
=P\left\{(N_{1})^{1/2}\widetilde{W}_{1}+\widetilde{m}_{1}(\boldsymbol{\sigma})\in E_{1},(N_{2})^{1/2}\widetilde{W}_{2}+
\widetilde{m}_{2}(\boldsymbol{\sigma})\in E_{2}
\right\}
\end{multline}
where 
\begin{equation}\label{nuove.variabili.aleatorie}
\widetilde{W}_{i}=\frac{W_{i}}{N_{i}}\qquad 
\widetilde{m}_{i}(\boldsymbol{\sigma})=\frac{\widetilde{S}_{i}(\boldsymbol{\sigma})}{N_{i}} \qquad i=1,2
\end{equation}
 while $E_{1}=(-\infty,\;(N_{1})^{-1/2}\theta_{1}]$ and $E_{2}=(-\infty,\;(N_{2})^{-1/4}\theta_{2}]$. Since $\mathbf{\widetilde{S}}$ is 
independent of
$(W_{1},W_{2})$, from equality (\ref{prob.1}) it follows that the distribution of the random vector (\ref{sum.vectors}) is the convolution 
of the distribution of $\left((N_{1})^{1/2}\widetilde{W}_{1},(N_{2})^{1/2}\widetilde{W}_{2}\right)$ with the distribution of 
$\mathbf{\widetilde{m}}=(\widetilde{m}_{1}(\boldsymbol{\sigma}),\widetilde{m}_{2}(\boldsymbol{\sigma}))$. By (\ref{bivariate.Gaussian}), the former distribution is:
\begin{equation*}
\dfrac{N\sqrt{\det\mathbf{\widetilde{J}}}}{2\pi} \;\exp\left(\!-\frac{N}{2}\langle\mathbf{\widetilde{J}}\mathbf{w},\mathbf{w}\rangle\right) 
\end{equation*}
while, by (\ref{measure.BG}), the latter is:
\begin{equation*}\label{distribution.sums.spins}
\frac{1}{Z_{N}}\exp\left(\frac{N}{2}\langle\mathbf{\widetilde{J}}\mathbf{m},\mathbf{m}\rangle\right)d\nu_{\mathbf{\widetilde{m}}}(\mathbf{m})
\end{equation*}
\noindent with $\nu_{\mathbf{\widetilde{m}}}(\mathbf{m})$ the distribution of $\mathbf{\widetilde{m}}$ on 
$(\mathbb{R}^{N},\prod_{i=1}^{N}\rho(\sigma_{i}))$. 
\noindent Thus:
\begin{align*}
P&\left\{(N_{1})^{1/2}\widetilde{W}_{1}+\widetilde{m}_{1}(\boldsymbol{\sigma})\in E_{1},(N_{2})^{1/2}\widetilde{W}_{2}+
\widetilde{m}_{2}(\boldsymbol{\sigma})\in E_{2}\right\}\nonumber\\
&=\dfrac{N\sqrt{\det\mathbf{\widetilde{J}}}}{2\pi Z_{N}}\!\!\iint_{E_{1}\times E_{2}\times\mathbb{R}^{2}}\!\!\!\!\!\!\!
\exp\bigg(\!\frac{N}{2}\Big(\!-
\langle\mathbf{\widetilde{J}}(\mathbf{w}-\mathbf{m}),(\mathbf{w}-\mathbf{m})\rangle+\langle\mathbf{\widetilde{J}}\mathbf{m},\mathbf{m}\rangle\Big)\bigg)d\nu_{\mathbf{\widetilde{m}}}(\mathbf{m})d\mathbf{w}
\nonumber\\
&=\dfrac{N\sqrt{\det\mathbf{\widetilde{J}}}}{2\pi Z_{N}}
\!\!\int_{E_{1}\times E_{2}}\!\!\!\!\!\exp\left(\!\!-\frac{N}{2}\langle\mathbf{\widetilde{J}}\mathbf{w},\mathbf{w}\rangle\right)\!\!
\int_{\mathbb{R}^{2}}\exp\left(N\langle\mathbf{\widetilde{J}}\mathbf{w},\mathbf{m}\rangle\right)d\nu_{\mathbf{\widetilde{m}}}
(\mathbf{m})d\mathbf{w}.
\end{align*}
\noindent By considering the definition of $\mathbf{\widetilde{S}}$ given in (\ref{definizione.S.tilde}), we can write the elements of 
$\mathbf{\widetilde{m}}$, given in (\ref{nuove.variabili.aleatorie}), in the following way:
\begin{align*}
\widetilde{m}_{1}(\boldsymbol{\sigma}) &= \dfrac{||\mathbf{v}_{M}||}{(\lambda_{M}-\lambda_{m})}
\left((\mathcal{H}_{11}-\lambda_{m})\dfrac{S_{1}(\boldsymbol{\sigma})}{N_{1}}+\mathcal{H}_{12}\dfrac{S_{2}(\boldsymbol{\sigma})}{N_{2}}
\right)\\
\widetilde{m}_{2}(\boldsymbol{\sigma}) &=\dfrac{||\mathbf{v}_{m}||}{(\lambda_{M}-\lambda_{m})}
\left((\lambda_{M}-\mathcal{H}_{11})\dfrac{S_{1}(\boldsymbol{\sigma})}{N_{1}}-\mathcal{H}_{12}\dfrac{S_{2}(\boldsymbol{\sigma})}{N_{2}}
\right)
\end{align*}
that allows to calculate:
\begin{align*}
\int_{\mathbb{R}^{2}}\exp\left(N\langle\mathbf{\widetilde{J}}\mathbf{w},\mathbf{m}\rangle\right)d\nu_{\mathbf{\widetilde{m}}}(\mathbf{m})&=\int_{\mathbb{R}^{N}}\exp\bigg(\frac{N\sum_{i\in P_{1}}\sigma_{i}}{N_{1}(\lambda_{M} -\lambda_{m})}\Big(||\mathbf{v}_{M}||(\mathcal{H}_{11}-\lambda_{m})(\widetilde{J}_{11}w_{1}\\
&\quad +\widetilde{J}_{12}w_{2})+||\mathbf{v}_{m}||(\lambda_{M}-\mathcal{H}_{11})(\widetilde{J}_{12}w_{1}+\widetilde{J}_{22}w_{2})\Big)\\
&\quad +\frac{N\mathcal{H}_{12}\sum_{i\in P_{2}}\sigma_{i}}{N_{2}(\lambda_{M}-\lambda_{m})}\Big(
||\mathbf{v}_{M}||(\widetilde{J}_{11}w_{1}+\widetilde{J}_{12}w_{2})\\
&\quad+||\mathbf{v}_{m}||(\widetilde{J}_{12}w_{1}+\widetilde{J}_{22}w_{2})\Big)\bigg)\prod_{i=1}^{N}d\rho(\sigma_{i})\\
&=\cosh^{N_{1}}\Big(\widetilde{a}_{1}w_{1}+\widetilde{a}_{2}w_{2}\Big)\cosh^{N_{2}}\left(\widetilde{b}_{1}w_{1}+\widetilde{b}_{2}
w_{2}\right)
\end{align*}
where:
\begin{align*}
\widetilde{a}_{1} &=\dfrac{\widetilde{J}_{11}||\mathbf{v}_{M}||(\mathcal{H}_{11}-\lambda_{m})+\widetilde{J}_{12}||\mathbf{v}_{m}||
(\lambda_{M}-\mathcal{H}_{11})}{\alpha(\lambda_{M}-\lambda_{m})}\\
\widetilde{a}_{2} &=\dfrac{\widetilde{J}_{12}||\mathbf{v}_{M}||(\mathcal{H}_{11}-\lambda_{m})+
\widetilde{J}_{22}||\mathbf{v}_{m}||(\lambda_{M}-\mathcal{H}_{11})}{\alpha(\lambda_{M}-\lambda_{m})}\\
\widetilde{b}_{1} &=\dfrac{\mathcal{H}_{12}(\widetilde{J}_{11}||\mathbf{v}_{M}||+\widetilde{J}_{12}||\mathbf{v}_{m}||)
}{(1-\alpha)(\lambda_{M}-\lambda_{m})}\\
\widetilde{b}_{2} &=\dfrac{\mathcal{H}_{12}(\widetilde{J}_{12}||\mathbf{v}_{M}||+
\widetilde{J}_{22}||\mathbf{v}_{m}||)}{(1-\alpha)(\lambda_{M}-\lambda_{m})}. 
\end{align*}
By computing the Euclidean norm of the two eigenvectors $\mathbf{v}_{M}$, $\mathbf{v}_{m}$ defined in (\ref{eigenvectors}) and considering the explicit expressions of $\widetilde{J}_{11}$, $\widetilde{J}_{12}$ and 
$\widetilde{J}_{22}$, given in (\ref{coefficenti.Jtilde}), it is easy to show that $\widetilde{a}_{i}=a_{i}$ and 
$\widetilde{b}_{i}=b_{i}$, $i=1,2$, where $a_{i}$ and $b_{i}$ for $i=1,2$, are defined in (\ref{coefficenti.abtilde}). 
Therefore after making the change of variable $x_{1}=(N_{1})^{1/2}w_{1}$,  $x_{2}=(N_{2})^{1/4}w_{2}$ and integrating over $\mathbf{s}$, we have:
\begin{align}\label{last.passage.step.1}
P\Big\{(N_{1})&^{1/2}\widetilde{W}_{1}+\widetilde{m}_{1}(\boldsymbol{\sigma})\in E_{1},(N_{2})^{1/2}\widetilde{W}_{2}+
\widetilde{m}_{2}(\boldsymbol{\sigma})\in E_{2}\Big\}\nonumber\\
&=\dfrac{N^{1/4}\sqrt{\det\mathbf{\widetilde{J}}}}{2\pi Z_{N}\alpha^{1/2}(1-\alpha)^{1/4}}
\int_{-\infty}^{\theta_{1}}\int_{-\infty}^{\theta_{2}}\exp\bigg(-\frac{N}{2}\bigg(\widetilde{J}_{11}
\left(\dfrac{x_{1}}{(N_{1})^{1/2}}\right)^{2}\nonumber\\
&\quad+\dfrac{2\widetilde{J}_{12}x_{1}x_{2}}{(N_{1})^{1/2}(N_{2})^{1/4}}+\widetilde{J}_{22}\left(\dfrac{x_{2}}{(N_{2})^{1/4}}\right)^{2}
\bigg)+N\alpha\ln\cosh\left(\dfrac{a_{1}x_{1}}
{(N_{1})^{1/2}}+\dfrac{a_{2}x_{2}}{(N_{2})^{1/4}}\right)\nonumber\\
&\quad+N(1-\alpha)\ln\cosh\left(\dfrac{b_{1}x_{1}}{(N_{1})^{1/2}}+\dfrac{b_{2}x_{2}}{(N_{2})^{1/4}}\right)\bigg)dx_{1}dx_{2}\nonumber\\
&=\dfrac{N^{1/4}\sqrt{\det\mathbf{\widetilde{J}}}}{2\pi Z_{N}\alpha^{1/2}(1-\alpha)^{1/4}}\int_{-\infty}^{\theta_{1}}
\int_{-\infty}^{\theta_{2}}\exp\left(-N\widetilde{G}\left(\dfrac{x_{1}}{(N_{1})^{1/2}},\dfrac{x_{2}}{(N_{2})^{1/4}}\right)\right)dx_{1}dx_{2}.
\end{align}

\noindent Taking $\theta_{1}\rightarrow\infty$ and $\theta_{2}\rightarrow\infty$ in the (\ref{last.passage.step.1}), we obtain an equation for 
$Z_{N}$ which when substituted back yields the result (\ref{step.1}). 

Now to conclude the proof of the theorem, by lemma \ref{lemma.1}, we have to analyze the 
distribution (\ref{step.1}) as $N\rightarrow\infty$, keeping in mind that only the first component of the random vector $(W_{1},W_{2})$
 contributes to the limit. Let us start by observing that the hypothesis on $G$ togheter with the definition 
$\widetilde{G}(\mathbf{x})=G(\mathbf{P}^{-1}\mathbf{x})$, imply that the origin is the unique minimum point of the function 
$\widetilde{G}$. Moreover, since the system of conditions (\ref{00.minimum.point}) is satisfied, the Hessian matrix of $\widetilde{G}$ computed in the origin,
$\mathbf{\mathcal{H}}_{\widetilde{G}}(0,0)=\diag\{\lambda_{M},\lambda_{m}\}$, has determinant equal to zero without being the null 
matrix. Thus, by Taylor expansion, there exists $\hat{\delta}>0$ sufficiently small so that, as 
$N\rightarrow\infty$, for $|x_{1}|<\hat{\delta} (N_{1})^{1/2}$ and 
$|x_{2}|<\hat{\delta} (N_{2})^{1/4}$ we can write:
\begin{equation}\label{taylor}
N\cdot \widetilde{G}
\left(\dfrac{x_{1}}{(N_{1})^{1/2}},\dfrac{x_{2}}{(N_{2})^{1/4}}\right)=\zeta_{1}x_{1}^{2}+\zeta_{2}x_{2}^{4}
+\sum_{|\boldsymbol{\eta}|=4\atop \eta_{1}\neq 0}\zeta_{\boldsymbol{\eta}}\frac{x_1^{\eta_1}x_2^{\eta_2}}{N^{\eta_1 /4}}+\sum_{|\boldsymbol{\eta}|=5}R_{\boldsymbol{\eta}}\frac{x_1^{\eta_1}x_2^{\eta_2}}{N^{(\eta_1+1) /4}}
\end{equation}
where $\boldsymbol{\eta}=(\eta_1,\eta_2)\in\mathbb{N}^{2}$ is a multi-index, $|\boldsymbol{\eta}|=\eta_1+\eta_2$, while the coefficients are the followings:
\begin{align}\label{coefficent.1}
\zeta_{1} &=\frac{1}{2!\alpha}\dfrac{\partial^{2}\widetilde{G}}{\partial x_{1}^{2}}(0,0) = \frac{\lambda_{M}}{2\alpha}\\
\zeta_{2} &=\frac{1}{4!(1-\alpha)}\dfrac{\partial^{4}\widetilde{G}}{\partial x_{2}^{4}}(0,0)= \dfrac{2\alpha(\alpha(1-\alpha J_{11})^{2}+
(1-\alpha)(1-(1-\alpha)J_{22})^{2})}{24\left(\alpha(1-\alpha J_{11})+(1-\alpha)(1-(1-\alpha)J_{22})\right)^{2}}\label{coefficent.2}\\
\zeta_{\boldsymbol{\eta}} &=\frac{\partial^{\boldsymbol{\eta}}\widetilde{G}(0,0)}{\boldsymbol{\eta}!\alpha^{\eta_1 /2}
(1-\alpha)^{\eta_2 /4}}\nonumber\\
 R_{\boldsymbol{\eta}} &=\int_{0}^{1}\frac{5(1-t)^4}{\boldsymbol{\eta}!\alpha^{\eta_1 /2}
(1-\alpha)^{\eta_2 /4}}
\partial^{\boldsymbol{\eta}}\widetilde{G}\left(\dfrac{tx_{1}}{(N_{1})^{1/2}},\dfrac{tx_{2}}{(N_{2})^{1/4}}\right)dt\nonumber
\end{align} 
with $\boldsymbol{\eta}!=\eta_{1}!\eta_{2}!$ and $\partial^{\boldsymbol{\eta}}=\partial^{|\boldsymbol{\eta}|}/\partial x_{1}^{\eta_{1}}\partial x_{2}^{\eta_{2}}$.
\noindent We observe that $\zeta_{1}$ and $\zeta_{2}$ are strictly positive because the model parameters fulfill the system of conditions (\ref{00.minimum.point}). Moreover, since the image under $\widetilde{G}$ of the origin is zero, we can find $\bar{\delta}>0$ sufficiently small so that, as 
$N\rightarrow\infty$, for $|x_{1}|<\bar{\delta} (N_{1})^{1/2}$ and 
$|x_{2}|<\bar{\delta} (N_{2})^{1/4}$
\begin{equation*}
\left|\sum_{|\boldsymbol{\eta}|=4\atop \eta_{1}\neq 0}\zeta_{\boldsymbol{\eta}}\frac{x_1^{\eta_1}x_2^{\eta_2}}{N^{\eta_1 /4}}+\sum_{|\boldsymbol{\eta}|=5}R_{\boldsymbol{\eta}}\frac{x_1^{\eta_1}x_2^{\eta_2}}{N^{(\eta_1+1) /4}}\right|\leq \frac{1}{2}(\zeta_{1}x_{1}^{2}+\zeta_{2}x_{2}^{4}).
\end{equation*}
Thus, defined $\delta=\min\{\hat{\delta},\bar{\delta}\}$, as 
$N\rightarrow\infty$, for $|x_{1}|<\delta (N_{1})^{1/2}$ and 
$|x_{2}|<\delta (N_{2})^{1/4}$ we have: 
\begin{align}\label{convergenza.dominata}
N\cdot \widetilde{G}
\left(\dfrac{x_{1}}{(N_{1})^{1/2}},\dfrac{x_{2}}{(N_{2})^{1/4}}\right)&\geq\zeta_{1}x_{1}^{2}+\zeta_{2}x_{2}^{4}
-\left|\sum_{|\boldsymbol{\eta}|=4\atop \eta_{1}\neq 0}\zeta_{\boldsymbol{\eta}}\frac{x_1^{\eta_1}x_2^{\eta_2}}{N^{\eta_1 /4}}+\sum_{|\boldsymbol{\eta}|=5}R_{\boldsymbol{\eta}}\frac{x_1^{\eta_1}x_2^{\eta_2}}{N^{(\eta_1+1) /4}}\right|\nonumber\\
&\geq \frac{1}{2}(\zeta_{1}x_{1}^{2}+\zeta_{2}x_{2}^{4}).
\end{align}
 Considered the set
$V=\{(x_{1},x_{2})\in\mathbb{R}^{2}:|x_{1}|\geq\delta (N_{1})^{1/2}, |x_{2}|\geq\delta (N_{2})^{1/4} \}$, by lemma \ref{lemma.3} there exists 
$\varepsilon >0$ such that for any bounded continuous function 
$\psi(\mathbf{x}):\mathbb{R}^{2}\rightarrow\mathbb{R}$:
\begin{equation}\label{limit.1}
\iint_{V}\exp\left(\!-N\widetilde{G}
\left(\dfrac{x_{1}}{(N_{1})^{1/2}},\dfrac{x_{2}}{(N_{2})^{1/4}}\right)\right)\psi(x_{1},x_{2})dx_{1}dx_{2}
=O\left(N^{3/4}e^{-N\varepsilon}\right).
\end{equation}
\noindent  On the other hand by (\ref{taylor}), (\ref{convergenza.dominata}) and dominate convergence, as $N\rightarrow\infty$:
\begin{multline}\label{limit.2}
\iint_{\mathbb{R}^{2}\setminus V} \!\!\!\exp\left(\!-N\widetilde{G}
\left(\dfrac{x_{1}}{(N_{1})^{1/2}},\dfrac{x_{2}}{(N_{2})^{1/4}}\right)\right)\psi(x_{1},x_{2})dx_{1}dx_{2}\\
\rightarrow\iint_{\mathbb{R}^{2}}\exp\left(\! -\zeta_{1}x_{1}^{2}-\zeta_{2}x_{2}^{4}\right)
\psi(x_{1},x_{2})dx_{1}dx_{2}.
\end{multline}
Therefore, by (\ref{limit.1}) and (\ref{limit.2}), as $N\rightarrow\infty$ we have that:
\begin{multline*}\label{step.2}
\dfrac{\displaystyle{\iint_{\mathbb{R}^{2}}}\!\exp\left(\!-N\widetilde{G}\left(\dfrac{x_{1}}{(N_{1})^{1/2}},\dfrac{x_{2}}{(N_{2})^{1/4}}
\right)\right)\psi(x_{1},x_{2})dx_{1}dx_{2}}{\displaystyle{\iint_{\mathbb{R}^{2}}}\!\exp\left(\!-N\widetilde{G}
\left(\dfrac{x_{1}}{(N_{1})^{1/2}},\dfrac{x_{2}}{(N_{2})^{1/4}}\right)\right)dx_{1}dx_{2}}\\
\rightarrow \dfrac{\displaystyle{\iint_{\mathbb{R}^{2}}}\!\exp\left(\! -\zeta_{1}x_{1}^{2}-\zeta_{2}x_{2}^{4}\right)\psi(x_{1},x_{2})
dx_{1}dx_{2}}{\displaystyle{\iint_{\mathbb{R}^{2}}}\!\exp\left(-\zeta_{1}x_{1}^{2}-\zeta_{2}x_{2}^{4}\right) dx_{1}dx_{2}}.
\end{multline*}


\noindent As mentioned previously, while $W_{2}$ does not contribute to the limit of the distribution (\ref{step.1}), the distribution
 obtained in the variable $x_{1}$, a Gaussian with zero mean and variance equal to $(2\zeta_{1})^{-1}$ where $\zeta_{1}$ is given in (\ref{coefficent.1}), is the convolution of the 
marginal distribution of $W_{1}$ with the limiting 
distribution of the first element of the vector (\ref{vettore.S.1}). Since the marginal distribution of $W_{1}$ is Gaussian too, if the difference 
$d$ between the variance of the distribution obtained by convolution and those of $W_{1}$ 
\begin{equation}\label{variance}
 d=\frac{\alpha}{\lambda_{M}}-\frac{\alpha \widetilde{J}_{22}}{\widetilde{J}_{11}\widetilde{J}_{22}-\widetilde{J}_{12}^{2}}
\end{equation}
is positive we can conclude that the limiting distributions of $(N_{1})^{-1/2}\widetilde{S}_{1}(\boldsymbol{\sigma})$ is a 
Gaussian with zero mean and variance equal to $d$.
To prove that $d$ is positive let us consider the strictly convex function 
$\Phi(\mathbf{x})=<\mathbf{\widetilde{J}}\mathbf{x},\mathbf{x}>-\widetilde{G}(\mathbf{x})$. After computing the second partial derivatives of $\Phi$ in the origin:
\begin{equation*}
\dfrac{\partial^{2}\Phi}{\partial x_{1}^{2}}(0,0) = \widetilde{J}_{11}-\lambda_{M},\quad
\dfrac{\partial^{2}\Phi}{\partial x_{1}\partial x_{2}} (0,0)=\widetilde{J}_{12},\quad
\dfrac{\partial^{2}\Phi}{\partial x_{1}^{2}} (0,0)=\widetilde{J}_{22},\\
\end{equation*}
and denoting the Hessian matrix of $\Phi$ by $\mathcal{H}_{\Phi}$, we can write: 
\begin{equation*}
 d=\dfrac{\alpha \det\mathcal{H}_{\Phi}(0,0)}{\lambda_{M}\det \mathbf{\widetilde{J}}}.
\end{equation*}
Since the function $\Phi$ is strictly convex and $\mathbf{\widetilde{J}}=\mathbf{P}^{-1}\mathbf{A}^{2}\mathbf{J}\mathbf{A}^{2}\mathbf{P}$ with 
$\mathbf{A}$ and $\mathbf{J}$ positive definite matrices and $\mathbf{P}$ an ortogonal matrix, we can conclude that $d>0$.\\
Thus the statement (\ref{result.theorem}) is proved by defining $\xi_{1}=(2d)^{-1}$ and $\xi_{2}=\zeta_{2}$ where $d$ is given by (\ref{variance}) and  $\zeta_{2}$ 
by (\ref{coefficent.2}). This concludes the proof of the theorem.

\section{Conclusions and Perspectives}
In this paper we extended previously obtained results (see \cite{fedele2011scaling}) on the limiting behavior of the random vector of total magnetizations for the 
bipartite mean-field model. We worked under the assumptions that the Hamiltonian is a convex function of the total magnetizations, the 
external field is away and the pressure functional admits a unique minimum point, the origin, in which the determinant of the Hessian 
matrix is equal to zero. As a consequence the homogeneity hypothesis on the pressure functional made in 
\cite{fedele2011scaling}, is true only if there is no interaction between particles of different groups, that is the bipartite 
mean-field model degenerates towards to distinct Curie-Weiss models. 
In the non-degenerate case, we found a non Gaussian limit distribution for the vector of the total magnetizations after being transformed with the orthogonal matrix that diagonalizes the Hessian matrix of the pressure functional. This result allows us to state that in the considered case the bipartite mean-field model undergoes a phase transition. 

The complete description of the asymptotic distribution of  the vector of total magnetizations both for the bipartite and the generic multipartite mean-field model will be subject of further investigations.   

\vspace{1cm}
\noindent {\bf Acknowledgments}: The author wishes to thank Professor C. Newman for interesting discussions and suggestions and 
F. Collet for helpful observations. The author also aknowledges the INdAM-COFUND Marie Curie fellowships for financial support.

\section*{Appendix A: proof of lemma \ref{lemma.2}}
\appendix
\noindent 	Considered the function
	\begin{equation*}
	\bar{G}(x_{1},x_{2},t_{1},t_{2})=\frac{1}{2}\left(\widetilde{J}_{11} x_{1}^{2}+2\widetilde{J}_{12}x_{1}x_{2}+\widetilde{J}_{22}x_{2}^{2}\right)-
	\alpha t_{1}(a_{1}x_{1}+a_{2}x_{2})
	-(1-\alpha)t_{2}(b_{1}x_{1}+b_{2}x_{2})
	\end{equation*}
	where $\mathbf{x}=(x_{1},x_{2})\in\mathbb{R}^{2}$ and $\mathbf{t}=(t_{1},t_{2})\in\{-1,1\}^{2}$ the following inequality holds:
	\begin{equation*}
	\widetilde{G}(\mathbf{x})\geq \min\{\bar{G}(\mathbf{x},\mathbf{t}):\mathbf{t}\in\{-1,1\}^{2}\}.
	\end{equation*}
	Thus:
	\begin{equation}\label{bo}
	\int_{\mathbb{R}^{2}}\exp\left(-\widetilde{G}(\mathbf{x})\right)d\mathbf{x}\leq\int_{\mathbb{R}^{2}}\exp\left(-\min\{\bar{G}(\mathbf{x},\mathbf{t}):
	\mathbf{t}\in\{-1,1\}^{2}\}\right)d\mathbf{x}.
	\end{equation}

	\noindent Since $\mathbf{\widetilde{J}}=\mathbf{P}^{-1}\mathbf{A}^{2}\mathbf{J}\mathbf{A}^{2}\mathbf{P}$ where the matrix $\mathbf{P}$ is 
	orthogonal while the matrices $\mathbf{J}$ and $\mathbf{A}$ are positive definite, the argument of the integral on 
	the right hand side of the inequality (\ref{bo}) is a Gaussian density function. This proves the statement (\ref{result.lemma.2}) for $N=1$. Now
	defined $m=\min\{\widetilde{G}(\mathbf{x}):\mathbf{x}\in\mathbb{R}^{2}\}$ and supposed true the inductive hypothesis:
	\begin{equation}\label{inductive.hypothesis}
	\int_{\mathbb{R}^{2}}\exp\left(-(N-1)\widetilde{G}(\mathbf{x})\right)d\mathbf{x}<\infty
	\end{equation}
	we have:
	\begin{align*}
	\int_{\mathbb{R}^{2}}\exp\left(-N\widetilde{G}(\mathbf{x})\right)d\mathbf{x}&=\int_{\mathbb{R}^{2}}\exp\left(-(N-1)\widetilde{G}(\mathbf{x})\right)
	\exp\left(-\widetilde{G}(\mathbf{x})\right)d\mathbf{x}\\
	&\leq e^{-m}\int_{\mathbb{R}^{2}}\exp\left(-(N-1)\widetilde{G}(\mathbf{x})\right)d\mathbf{x}
	\end{align*}
	where the latter integral is finite by the inductive hypothesis (\ref{inductive.hypothesis}). This proves the result 
(\ref{result.lemma.2}) for any $N\in\mathbb{N}\setminus\{0\}$.

\section*{Appendix B: proof of lemma \ref{lemma.3}}
\appendix
\noindent 	Since the set $V$ contains no global minima of $\widetilde{G}(\mathbf{x})$, there exists $\varepsilon>0$ such that:
	\begin{equation*}
	\inf\{\widetilde{G}(\mathbf{x}):\mathbf{x}\in V\}\geq\inf\{\widetilde{G}(\mathbf{x}):\mathbf{x}\in\mathbb{R}^{2}\}+\varepsilon
	=m+\varepsilon.
	\end{equation*}
	Therefore we can write:
	\begin{align*}
	e^{Nm}\int_{V}\exp(-N\widetilde{G}(\mathbf{x}))d\mathbf{x} &<e^{Nm}e^{-(N-1)(m+\varepsilon)}\int_{V}
	\exp(-\widetilde{G}(\mathbf{x}))d\mathbf{x}\nonumber\\
	&\leq e^{-N\varepsilon}\left(e^{(m+\varepsilon)}\int_{\mathbb{R}^{2}}\exp(-\widetilde{G}(\mathbf{x}))d\mathbf{x}\right)
	\end{align*}
	where the latter integral is finite by lemma \ref{lemma.2}. Thus the statement (\ref{result.lemma.3}) is proved.

\end{document}